\documentclass[review]{elsarticle}
\usepackage{hyperref} 
\usepackage{hyperref} 
\usepackage{mathrsfs}
\usepackage{bm}
\usepackage{titlesec}
\usepackage[thinlines,thiklines]{easybmat}
\usepackage{graphicx}
\usepackage{float}
\usepackage{color}
\usepackage{braket}
\usepackage{xcolor}
\usepackage{tablefootnote}
\usepackage{graphicx}
\usepackage{amsmath}
\usepackage{amssymb}
\usepackage{amsthm}
\usepackage{amsfonts}
\usepackage{tcolorbox}
\usepackage{colortbl}
\usepackage{mathrsfs}
\usepackage{titlesec}
\usepackage[thinlines,thiklines]{easybmat}
\usepackage{graphicx}
\usepackage{float}
\usepackage{color}
\usepackage{braket}
\usepackage{xcolor}
\usepackage{tablefootnote}
\usepackage{graphicx}
\usepackage{amsmath}
\usepackage{amssymb}
\usepackage{amsthm}
\usepackage{amsfonts}
\usepackage{tcolorbox}
\usepackage{colortbl}
\theoremstyle{plain}\newtheorem{theorem}{Theorem}
\newtheorem{Lemma}{Lemma}

\newtheorem{Definition}{Definition}

\newtheorem{corollary}{Corollary}
\newtheorem{Fact}{Fact}
\newtheorem{remark}{Remark}

\usepackage{algorithm}
\usepackage{algorithmic}
\newcommand{\inbrace}[1]{\left \{ #1 \right \}}
\newcommand{\inparen}[1]{\left ( #1 \right )}

\begin{document}
\begin{frontmatter}
\author[a1]{Zekun Ye}
\author[a1,a2]{Lvzhou Li\corref{one}}

\cortext[one]{Corresponding author. 
 \indent {\it E-mail
address:} yezekun@mail2.sysu.edu.cn (Z. Ye); lilvzh@mail.sysu.edu.cn (L.
Li)}

\address[a1]{Institute of Quantum Computing and Computer Theory, School of Computer
Science and Engineering, Sun Yat-sen University, Guangzhou 510006, China}

\address[a2]{ Key Laboratory of Machine Intelligence and Advanced Computing, Ministry of Education, China}

\title{Deterministic Algorithms for the Hidden Subgroup Problem}



\begin{abstract}
We consider deterministic algorithms for the well-known hidden subgroup problem ($\mathsf{HSP}$): for a finite group $G$ and a finite set $X$, given a function $f:G \to X$ and the promise that for any $g_1, g_2 \in G, f(g_1) = f(g_2)$ iff $g_1H=g_2H$ for a subgroup $H \le G$, the goal of the decision version is to determine whether $H$ is trivial or not, and the goal of the identification version is to identify $H$. An algorithm for the problem should query $f(g)$ for $g\in G$ at least as possible.
Nayak \cite{Nayak2021} asked whether there exist deterministic algorithms with $O(\sqrt{\frac{|G|}{|H|}})$ query complexity for $\mathsf{HSP}$. We answer this problem by proving the following results, which also extend the main results of Ref. [30], since here the algorithms do not rely on any prior knowledge of $H$.
 
\begin{itemize}
	\item [(i)] 
	When $G$ is a general finite Abelian group, there exist an algorithm with $O(\sqrt{\frac{|G|}{|H|}})$ queries to decide the triviality of $H$ and an algorithm to identify $H$ with $O(\sqrt{\frac{|G|}{|H|}\log |H|}+\log |H|)$ queries.

	\item [(ii)] In general there is no deterministic algorithm for the identification version of $\mathsf{HSP}$ with query complexity of $O(\sqrt{\frac{|G|}{|H|}})$, since there exists an instance of $\mathsf{HSP}$ that needs $\omega(\sqrt{\frac{|G|}{|H|}})$ queries to identify $H$.\footnote{$f(x)$ is said to be $\omega(g(x))$ if for every positive constant $C$, there exists a positive constant $N$ such that for $x>N$, $f(x)\ge C\cdot g(x)$, which means $g$ is a strict lower bound for $f$.} On the other hand, there exist instances of $\mathsf{HSP}$ with query complexity far smaller than $O(\sqrt{\frac{|G|}{|H|}})$, whose query complexity is
$O(\log \frac{|G|}{|H|})$ and even $O(1)$. 
\end{itemize}

\end{abstract}

\end{frontmatter}

\section{Introduction}
\subsection{Background}

The hidden subgroup problem plays an important role in the history of quantum computing. 
Several important quantum algorithms such as Deutsch-Jozsa algorithm \cite{deutsch1992rapid}, Simon algorithm \cite{simon1994power}, and Shor algorithm \cite{shor1994algorithms} have a uniform description in the framework of the hidden subgroup problem \cite{jozsa1998quantum}. Moreover, many quantum algorithms were proposed for the instances of the hidden subgroup problem, e.g., \cite{bacon2005from,childs2007quantum,ettinger2004the,grigni2001quantum,hallgren2003the,kempe2005the,kuperberg2005a}.

The hidden subgroup problem is formally defined as follows.
\begin{Definition} \label{df_hsp}[The hidden subgroup problem]
\
\begin{tcolorbox}
			\label{identification version}
			\noindent
			\textbf{Given:} A finite group $G$; an (unknown) function $f:G \to X$, where $X$ is a finite set.
			\\
			\\
			\textbf{Promise:} There exists a subgroup $H$ such that for any $g_1, g_2 \in G, f(g_1) = f(g_2)$ iff $g_1H=g_2H$.
			\\
			\\
			\textbf{Problem:} Identify $H$.
		\end{tcolorbox}
\end{Definition}

The hidden subgroup problem consists of the Abelian hidden subgroup problem and the non-Abelian hidden subgroup problem according to whether $G$ is 
commutative or not. Many problems are special cases of the Abelian hidden subgroup problem, such as Simon's problem \cite{simon1994power}, generalized Simon's problem \cite{Ye2019query} and some important number-theoretic problems \cite{Hallgren2007Polynomial,Hallgren2005Fast,Schmidt2005Polynomial}. The non-Abelian hidden subgroup problem also received much attention \cite{grigni2001quantum,IvanyosMS01,kuperberg2005a,FriedlIMSS03,hallgren2003the,MooreRRS07,Regev04}. While there exist efficient quantum algorithms to solve the Abelian hidden subgroup problem \cite{simon1994power,BonehL95,Kitaev1995quantum,brassard1997exact,MoscaE98,EisentragerHK014}, many instances of the non-Abelian hidden subgroup problem are not known to have efficient quantum algorithms, such as the dihedral hidden subgroup problem and the symmetric hidden subgroup problem \cite{kuperberg2005a, Decker2013Hidden}. 


 In this paper, we focus on deterministic algorithms for the hidden subgroup problem.
 Deterministic and randomized algorithms are called classical algorithms. A deterministic algorithm for the above problem usually computes $f(g)$ for enough amount of $g\in G $ to find $H$ (or check whether $H$ is trivial or not). The number of $g\in G $ computed for $f$ is called the query complexity of the algorithm. For example, a naive algorithm with query complexity $O(|G|)$ is to compute $f(g)$ for all $g\in G $. A nontrivial problem worthy of consideration is: how difficult is this problem in terms of query complexity and how to find an optimal query algorithm for it. 
There exist a lot of discussions about the classical query complexity for some instances of the hidden subgroup problem. For example, the classical query complexity of Simon's problem was proven to be $\Theta(\sqrt{2^n})$ \cite{simon1994power,cai2018optimal,Wolf2019quantum}. Ye et al. \cite{Ye2019query} proved that a nearly optimal bound for the classical query complexity of generalized Simon's problem ($\mathsf{GSP}$). For the order-finding problem over $\mathbb{Z}_{2^m} \times \mathbb{Z}_{2^n}$, Cleve \cite{Cleve04} proved that the deterministic query complexity is at least $\Omega\left(\sqrt{\frac{2^n}{m}}\right)$, and the randomized query complexity is at least $\Omega\left(\frac{2^{n/3}}{\sqrt{m}}\right)$. Kuperberg \cite{kuperberg2005a} proved the classical query complexity of the dihedral hidden subgroup problem over the dihedral group $D_n$ is $\Omega(\sqrt{N})$. Childs \cite{Childs2021Lecture} showed that a classical algorithm must make $\Omega(\sqrt{N})$ queries if there are $N$ candidate subgroups whose only common element is the identity element. 

Recently, Nayak \cite{Nayak2021} considered deterministic algorithms for the hidden subgroup problem. He gave a deterministic algorithm with $O(\sqrt{|G|})$ queries to identify $H$ for the Abelian hidden subgroup problem, and showed there exists a deterministic algorithm with query complexity $O(\sqrt{|G|\ln |G|})$ for the general hidden subgroup problem. Finally, he proposed an open problem: is there a deterministic algorithm to the hidden subgroup problem with query complexity $O(\sqrt{\frac{|G|}{|H|}})$? 
\subsection{ Our contributions}\label{sec:resut}

 In this paper, we will answer the above open problem.
In the following, the hidden subgroup problem will be abbreviated to $\mathsf{HSP}$. The one in Definition \ref{df_hsp} will be called the identification version of $\mathsf{HSP}$. We will also consider its decision version where the goal is to decide whether the hidden subgroup $H$ is trivial or not, instead of finding $H$.


Our main results are the following theorems and corollaries, which not only answer the open problem proposed by Nayak \cite{Nayak2021}, but also extend the main results of Ref. \cite{Ye2019query}. Theorem \ref{th:simple_case} tells that there exist instances of $\mathsf{HSP}$ with query complexity far smaller than $O(\sqrt{\frac{|G|}{|H|}})$. Theorem \ref{th:general_case} shows that there exist deterministic algorithms with query complexity $O(\sqrt{\frac{|G|}{|H|}})$ for the decision version of the Abelian hidden subgroup problem and deterministic algorithms with query complexity not far from $O(\sqrt{\frac{|G|}{|H|}})$ for the identification version. Theorem \ref{th:lower_bound} and its corollary indicate that in general there is no deterministic algorithm with query complexity $O(\sqrt{\frac{|G|}{|H|}})$ for the identification version of $\mathsf{HSP}$. In addition, Corollary \ref{co:gsp} extends the main results of Ref. \cite{Ye2019query}, since the algorithm in Ref. \cite{Ye2019query} requires to know in advance the rank of the hidden group $H$, whereas the algorithms in this paper do not rely on any prior knowledge of $H$.

\begin{theorem}
\label{th:simple_case}
For $G = {\mathbb{Z}}_{p^k}$, where $p$ is a prime, there exist a deterministic algorithm using 2 queries to solve the decision version of $\mathsf{HSP}$, and a deterministic algorithm using $O(\log \frac{|G|}{|H|})$ queries to solve the identification version.
\end{theorem}

\begin{theorem}\label{th:general_case}
There exist a deterministic algorithm using at most $O(\sqrt{\frac{|G|}{|H|}})$ queries to solve the decision version of $\mathsf{HSP}$ over a general finite Abelian group $G$, and a deterministic algorithm using at most $O(\sqrt{\frac{|G|}{|H|}\log H}+\log |H|)$ queries to solve the identification version.
\end{theorem} 

By Theorem \ref{th:general_case}, we have the following corollary.

\begin{corollary}
\label{co:gsp}
For $G = \mathbb{Z}_p^n$, where $p$ is a prime, there exist a deterministic algorithm using at most $O(\sqrt{p^{n-k}})$ queries to decide whether $H$ is trivial or not, and a deterministic algorithm using at most $O(\max\{k,\sqrt{k \cdot p^{n-k}}\})$ queries to identify $H$, where $k = \log_p |H|$.
\end{corollary}

\begin{proof}
First, $k = \log_p |H|$ is equivalent to $|H| = p^k$.
Then substituting $|G| = p^n$ and $|H| = p^k$ into Theorem \ref{th:general_case}, we have $O(\sqrt{\frac{|G|}{|H|}}) = O(\sqrt{p^{n-k}})$ and 
$$
\begin{aligned}
O(\sqrt{\frac{|G|}{|H|}\log |H|}+\log |H|) &= O(\sqrt{p^{n-k} \cdot k}+k) \\
&= O(\max\{k,\sqrt{k \cdot p^{n-k}}\}),
\end{aligned}
$$
where the last equation follows from the fact that for $a,b \ge 0$, since $\frac{a+b}{2} \le \frac{a+b}{2}+\frac{|a-b|}{2} = \max\{a,b\} \le a+b$, we have $\Theta\inparen{\max\{a,b\}} = \Theta\inparen{a+b}$.  
\end{proof}

Note that Ref. \cite{Ye2019query} also gave an identification algorithm with the same query complexity as Corollary \ref{co:gsp}. The difference between them is that the algorithm in Ref. \cite{Ye2019query} requires that $k$ is known in advance, whereas our algorithms do not need to know $k$ initially. Thus, Corollary \ref{co:gsp} is a stronger conclusion and extend Theorem 3 of Ref. \cite{Ye2019query}.
 
\begin{theorem}
\label{th:lower_bound}
Any deterministic algorithm needs $\Omega(\frac{\log |\mathcal{H}|}{\log \frac{|G|}{|H|}})$ queries to solve the identification version of $\mathsf{HSP}$, where $\mathcal{H} = \{H' \le G| |H'| = |H|\}$.
\end{theorem}

It is easy to get the following corollary from Theorem \ref{th:lower_bound}. It was also implied by Lemma 2 of \cite{Ye2019query}.

\begin{corollary}
\label{pro:lower_bound}
Any deterministic algorithm needs $\Omega(\log |H|)$ queries to solve the identification version of $\mathsf{HSP}$ over $G = \mathbb{Z}_p^n$.
\end{corollary}

\begin{proof}
For $G = \mathbb{Z}_p^n$, suppose $H$ is a subgroup of rank $k$ in $G$. Then $|H| = p^k$. 
By Ref. \cite{Ye2019query},
$$
|\mathcal{H}| = \prod_{j=0}^{k-1} \frac{p^{n}-p^j}{p^{k}-p^j} > p^{(n-k)k}.
$$
Thus, by Theorem \ref{th:general_case}, a lower bound on the query complexity is
$$
\Omega\inparen{\frac{\log |\mathcal{H}|}{\log \frac{|G|}{|H|}}} = \Omega\inparen{\frac{\log p^{(n-k)k}}{\log p^{n-k}}} = \Omega\inparen{k} = \Omega\inparen{\log |H|}.
$$
\end{proof}

\begin{remark} If $G = \mathbb{Z}_p^n$ and $|H| = p^{n-1}$, then any deterministic algorithm needs $\Omega(n)$ queries to identify $H$ by Corollary \ref{pro:lower_bound}. Note that in this case we have $\sqrt{\frac{|G|}{|H|}} = O(1)$. Therefore, any deterministic algorithm for the problem requires $\omega(\sqrt{\frac{|G|}{|H|}})$ queries to identify $H$.
\end{remark}
\subsection{Organization}
	The remainder of the paper is organized as follows. In Section \ref{sec:Pre}, we review some notations used in this paper.
	In Section \ref{sec:simple_case}, two deterministic algorithms are presented for solving $\mathsf{HSP}$ over $G = {\mathbb{Z}}_{p^k}$. 
	In Section \ref{sec:general_case}, algorithms are designed to solve $\mathsf{HSP}$ over general finite Abelian groups. In Section \ref{sec:lower_bound}, a general lower bound is given for the query complexity of the identification version of $\mathsf{HSP}$. Finally, a conclusion is made in Section \ref{sec:Conclusion}.

\section{Preliminary}
\label{sec:Pre}
In this section, we present some notations used in this paper. Let $[m] = \{1,2,...,m\}$
and $\mathbb{Z}_{p^k}$ denote the additive group of elements $\{0,1,...,p^k-1\}$ with addition modulo $p^k$.
For two groups $G_1,G_2$, let $G_1 \times G_2$ denote the direct product of $G_1$ and $G_2$.
For an arbitrary finite Abelian group $G$, we have $G \cong \{0\} \times \mathbb{Z}_{p_1^{k_1}} \times \cdots \times \mathbb{Z}_{p_l^{k_l}}$ \footnote{Generally, we say $G \cong \mathbb{Z}_{p_1^{k_1}} \times \cdots \times \mathbb{Z}_{p_l^{k_l}}$. In this paper, we add \{0\} for the simplicity of algorithms discription.}. For $a,b \in G$, suppose $a = (0,a_1,...,a_l),b=(0,b_1,...,b_l)$. Then we define the group operation ``+" as follows:
$$
a+b = (0,a_1+b_1\mod p_1^{k_1},...,a_l+b_l\mod p_l^{k_l}). \\
$$
Correspondingly, we define 
$$
a-b = (0,a_1-b_1\mod p_1^{k_1},...,a_l-b_l\mod p_l^{k_l}).
$$
For $W_1,W_2 \subseteq G$, we have
$$
\begin{aligned}
W_1+W_2 &= \{a+b| a \in W_1, b\in W_2\}, \\
W_1-W_2 &= \{a-b|a \in W_1, b \in W_2\}. \\
\end{aligned}
$$
  $W_1,W_2$ is called a \textit{generating pair} of $V$ if $V = W_1-W_2$.  $W_1,W_2$ is called to have  a collision if there exist $x \in W_1,y \in W_2$ such that $f(x) = f(y)$.

We denote that $H$ is a subgroup of $G$ by $H \le G$. For $H \le G$, a subset $S$ is said to be a generating set for $H$ if all elements in $H$ can be expressed as the finite sum of elements in $S$ and their inverses, i.e., $H = \langle S \rangle = \{{l_1}a_1+{l_2}a_2+\cdots+{l_n}a_n|a_i \in S, l_i = \pm 1, n\in \mathbb{N}\}$. For $g \in G$, we have $\langle g \rangle = \{\underbrace{g+\cdots+g}_{n}|n\in \mathbb{Z}\}$ and $gH = \{g+h|h\in H\}$. If $g$ is a nonzero element, we let $I_g$ denote the maximum nonzero coordinate. For example, if $g = (0,g_1,...,g_i,0,...,0)$ and $g_i \neq 0$, then $I_g = i$. If $g = 0$, then we let $I_g = 0$.  
For $g_1,g_2 \in G$, $g_1H = g_2H$ is equivalent to $g_1-g_2 \in H$. 
The group with only one element, the identity element, is called the \textit{trivial group}. 
 


Additionally, we use $W_1 \times \cdots \times W_i$ to represent $W_1 \times \cdots \times W_i \times \underbrace{\{0\} \cdots \times \{0\}}_{l-i+1}$ for simplicity. In the description and analysis of Algorithm \ref{al:main1} and \ref{al:main2}, let $G_i = \{0\} \times \mathbb{Z}_{p_1^{k_1}} \times \cdots \times \mathbb{Z}_{p_i^{k_i}}$, $H_i = \{h\in H|I_h \le i\}$.  
	
\section{Algorithms for $\mathsf{HSP}$ over $G = {\mathbb{Z}}_{p^k}$}
\label{sec:simple_case}
In this section, we propose  algorithms to solve the hidden subgroup problem over $G = {\mathbb{Z}}_{p^k}$ ($p$ is a prime), which establishes  Theorem \ref{th:simple_case}. 
 
\subsection{Decision version}
 The algorithm of the decision version is as follows: Query $f(0)$ and $f(p^{k-1})$. If $f(0) = f(p^{k-1})$, then $H$ is non-trivial; otherwise, $H$ is trivial. The correctness of the algorithm relies on the following fact.
	\begin{Fact}
	The subgroups of ${\mathbb{Z}}_{p^k}$ consist of $\{0\}, \langle p^{k-1} \rangle, \langle p^{k-2} \rangle,..., \langle 1 \rangle$.  
	\end{Fact}
Since all non-trivial subgroups of $G$ contain the element $p^{k-1}$, we only need to check whether $p^{k-1} \in H$ or not to decide the triviality of $H$.

\subsection{Identification version}
In the following, we give an algorithm to identify $H$ (Algorithm \ref{al:simple2}). The idea of Algorithm \ref{al:simple2} is to check whether $p^i \in H$ for $0 \le i \le k-1$. If there exists some $i$ such that $p^i \in H$, then $H = \langle p^i \rangle$; otherwist, $H$ is trivial.

	\begin{algorithm}
	\caption{Find $H$ in $\mathbb{Z}_{p^k}$}
	\label{al:simple2}
		\begin{algorithmic}[1]
			\STATE Query $0$;
			\FOR{$i=0 \to k-1$}
				\STATE Query $p^i$;
				\IF{$f(0) = f(p^i)$}
					\RETURN $\langle p^i \rangle$;
				\ENDIF
			\ENDFOR
			\RETURN $\{0\}$;
		\end{algorithmic}
	\end{algorithm}
	
Specifically, we analyze the correctness and query complexity of Algorithm \ref{al:simple2} as follows. If $H = \langle p^j \rangle$ for some $0 \le j \le k-1$, then 
 $f(0) = f(x)$ iff $x \in H$, i.e., $p^j|x$. Thus, the algorithm will perform the loop repeatly for $0 \le i \le j$. Finally, the algorithm returns $\langle p^j \rangle$. 
 Since $|G| = p^k, |H| = p^{k-j}$, the query complexity of the algorithm is $1+(j+1) = O(\log \frac{|G|}{|H|})$.

If $H = \{0\}$, then the algorithm will perform the loop entirely since $f(0) \neq f(p^i)$ for any $0 \le i \le k-1$. Finally, the algorithm returns $\{0\}$. Since $|H| = 1$, the query complexity of the algorithm is $1+k = O(\log \frac{|G|}{|H|})$.

\section{General algorithms for $\mathsf{HSP}$ over finite Abelian groups}\label{sec:general_case}
In this section, we will discuss the general Abelian hidden subgroup problem. Since any finite Abelian group is isomorphic to $\{0\} \times \mathbb{Z}_{p_1^{k_1}} \times \cdots \times \mathbb{Z}_{p_l^{k_l}}$, where $p_i$'s are primes and $k_i \ge 1$ for any $i$, we only need to consider the case $G = \{0\} \times \mathbb{Z}_{p_1^{k_1}}\times \cdots \times \mathbb{Z}_{p_l^{k_l}}$. The decision and identification versions are discussed in Section \ref{sec:general_case_decision} and \ref{sec:general_case_identification}, respectively.
The query complexities of our algorithms are related to $H$, even if we do not know any information about $H$ initially. 
\subsection{Finding generating pairs}
We first give Algorithm \ref{al:subroutine} as follows, which is a subroutine of the main algorithms solving $\mathsf{HSP}$ (i.e., Algorithm \ref{al:main1} and \ref{al:main2}). The goal of Algorithm \ref{al:subroutine} is to find a generating pair $W_1,W_2$ such that $W_1-W_2 = V$, $|W_1| \le 2\lceil\sqrt{|V|\cdot r}\rceil$, $|W_2| \le \lceil\sqrt{|V|/r}\rceil$. This subroutine does not make queries. 


\begin{algorithm}
\caption{findPair}
\label{al:subroutine}
\begin{algorithmic}[1]
\REQUIRE $V = \{0\} \times \mathbb{Z}_{p_1^{j_1}} \times \mathbb{Z}_{p_2^{j_2}} \times \cdots \times \mathbb{Z}_{p_l^{j_l}}$ and $r \in [1,l)$, where $p_i$'s are primes and $j_i \ge 0$ for any $i$;
\ENSURE $W_1,W_2$ s.t. $W_1-W_2 = V$, $|W_1| \le 2\lceil\sqrt{|V|\cdot r}\rceil$, $|W_2| \le \lceil\sqrt{|V|/r}\rceil$.

\STATE Select $i$ such that $\prod_{m=i+1}^l p_m^{j_m} < \lceil\sqrt{|V|/r}\rceil \le \prod_{m=i}^l p_m^{j_m}$;
 \STATE Let $a = \lfloor \frac{\lceil\sqrt{|V|/r}\rceil}{\prod_{m=i+1}^l p_m^{j_m}} \rfloor$, $b = \lceil \frac{p_i^{j_i}}{a} \rceil$;  
\STATE $W_1 = \{0\} \times \mathbb{Z}_{p_1^{j_1}} \times \cdots \times \mathbb{Z}_{p_{i-1}^{j_{i-1}}} \times \{0,1,...,b-1\} \times \{0\} \times \cdots \times \{0\}$;
\STATE $W_2 = \{0\} \times \cdots \times \{0\} \times \{0,-b,...,-(a-1)b\} \times \mathbb{Z}_{p_{i+1}^{j_{i+1}}} \times \cdots \times \mathbb{Z}_{p_{l}^{j_{l}}}$.  
\RETURN $W_1,W_2$.
		\end{algorithmic}
	\end{algorithm}

In the following, we show that the above returned values are satisfied. Since $a \ge \frac{\lceil\sqrt{|V|/r}\rceil}{\prod_{m=i+1}^l p_m^{j_m}}$, we have
$\prod_{m=i+1}^l p_m^{j_m} \ge \frac{\lceil\sqrt{|V|/r}\rceil}{a}$. Since $\lceil\sqrt{|V|/r}\rceil \le \prod_{m=i}^l p_m^{j_m}$, we have 
$$
\prod_{m=1}^{i-1} p_m^{j_m} = \frac{|V|}{\prod_{m=i}^l p_m^{j_m}}\le \frac{|V|}{\lceil\sqrt{|V|/r}\rceil} \le \sqrt{|V|\cdot r}.
$$ 
Thus,

$$
\begin{aligned} 
 |W_1| &= \prod_{m=1}^{i-1}p_m^{j_m} \cdot b \\
 &= \prod_{m=1}^{i-1}p_m^{j_m}\cdot \lceil \frac{p_i^{j_i}}{a} \rceil \\
 &\le\prod_{m=1}^{i-1}p_m^{j_m}+\prod_{m=1}^{i-1}p_m^{j_m} \cdot \frac{p_i^{j_i}}{a} \\
&\le \lceil\sqrt{|V|\cdot r}\rceil+\frac{|V|}{p_i^{j_i} \prod_{m=i+1}^l p_m^{j_m}} \cdot \frac{p_i^{j_i}}{a} \\
&\le \lceil\sqrt{|V|\cdot r}\rceil
+ \frac{|V|}{p_i^{j_i} \frac{\lceil\sqrt{|V|/r\rceil}}{a}} \cdot \frac{p_i^{j_i}}{a} \\
&\le 2\lceil\sqrt{|V|\cdot r}\rceil
\end{aligned}
$$
and 
$$
\begin{aligned}
W_2 &= a \prod_{m=i+1}^l p_m^ {j_m} \\
&\le \frac{\lceil\sqrt{|V|/r}\rceil}{\prod_{m=i+1}^l p_m^{j_m}}\prod_{m=i+1}^l p_m^ {j_m} \\
&= \lceil\sqrt{|V|/r}\rceil. \\
\end{aligned}
$$
Moreover, since $ab \ge p_i^{j_i}$, we have $\{0,1,...,b-1\}-\{0,b,...,-(a-1)b\} = \{0,1,...,ab-1\} \equiv \mathbb{Z}_{p_i^{j_i}}$ $(mod\ p_i^{j_i})$. Thus, $W_1 - W_2 = V$.

\subsection{Decision version}
\label{sec:general_case_decision}
 

In the following, we give Algorithm \ref{al:main1} to solve the decision version of Abelian hidden subgroup problem. 
The idea  is to verify whether $G_i$ has nonzero elements in $H$ for $1 \le i \le l$. If for some $i$ the answer is yes, then $H$ is nontrivial; otherwise, $H$ is trivial. 
We use the following observation: in order to verify whether a set has nonzero elements in $H$, it suffices to check whether a generating pair of the set have collisions.

\begin{algorithm}
	\caption{The algorithm for the decision version}
	\label{al:main1}
	\begin{algorithmic}[1]
	\REQUIRE $G \cong \{0\} \times \mathbb{Z}_{p_1^{k_1}} \times \cdots \times \mathbb{Z}_{p_l^{k_l}}$, where $p_i$'s are primes and $k_i \ge 1$ for any $i$;
	\ENSURE whether $H$ is trivial or not;
		\STATE $W_1 = W_2 = \{0\}$; \label{main1_init}
		\FOR{$i=1 \to l$} \label{main1_for1}
			\STATE Query all not queried elements in $W_1 \times \{0\}$, $W_2 \times \{p_i^{k_i-1}\}$; \label{main1_query}
				\IF{there exist $x \in W_1 \times \{0\}$, $y \in W_2 \times \{p_i^{k_i-1}\}$ s.t. $f(x) = f(y)$} \label{main1_if}
					\RETURN false;
				\ENDIF 
      \STATE $W_1,W_2 \leftarrow$ findPair$(G_i,1)$; \label{main1_update1}
		 \ENDFOR
		 \RETURN true;
		\end{algorithmic}
	\end{algorithm}

Specifically, the process of Algorithm \ref{al:main1} is as follows. In Step \ref{main1_init}, we initialize $W_1 = W_2 = \{0\}$. Then we go into the loop. We query all the elements  not queried in $W_1 \times \{0\}$ and $W_2 \times \{p_i^{k_i-1}\}$. If there exist $x \in W_1 \times \{0\}$, $y \in W_2 \times \{p_i^{k_i-1}\}$ such that $f(x) = f(y)$, then the algorithm returns false; otherwise we update the values of $W_1,W_2$ and go to the next iteration. After performing $l$ iterations, if the algorithm does not return false, then it returns true.  

We first analyze the correctness of Algorithm \ref{al:main1}. Let $d = \min_{h \in H} I_h$. We discuss the following two cases: i) $H$ is trivial. In this case, $f(x)\not= f(y)$ for any $x \neq y$, and thus Algorithm \ref{al:main1} does not return false at each iteration. Finally, Algorithm \ref{al:main1} returns true. ii) $H$ is non-trivial. In this case, if the algorithm is over after no more than $d-1$ iterations of the loop, then the algorithm must return false, which is a correct result; if the algorithm is not over after $d-1$ iterations, then $W_1 - W_2 = G_{d-1}$. Since $d = \min_{h \in H} I_h$, there exists $h \in H$ such that $h = (0,h_1,...,h_d,0,...,0)$, where $h_d \neq 0$. Then there exists $b$ such that $bh_d = p_d^{k_d-1}$ by Lemma \ref{lemma:prime}, i.e., $bh = (0,bh_1,...,bh_{d-1},p_d^{k_d-1},0,...,0)$. Thus, 
$$
bh \in G_{d-1} \times \{p_d^{k_d-1}\} = W_1\times \{0\}-W_2\times \{p_d^{k_d-1}\}.
$$
Since $bh \in H$, for some $x, y$ satisfying that $x-y = bh$, we will find $f(x) = f(y)$ in Step \ref{main1_if} of the $d$-th iteration. Hence, Algorithm \ref{al:main1} will return false. Totally, Algorithm \ref{al:main1} always outputs the correct result.

\begin{Lemma}
\label{lemma:prime}
For any nonzero element $h \in \mathbb{Z}_{p^k}$, there exists an element $b \in \mathbb{Z}_{p^k}$ such that $bh \equiv p^{k-1} (mod\ p^k)$.
\end{Lemma}

\begin{proof}
Let $h' = h\mod p$. Since $h' \in \{1,...,p-1\}$, there exists $a \in \mathbb{Z}_p$ such that $ah' \equiv 1 (mod\ p)$. Thus, $ah'p^{k-1} \equiv p^{k-1} (mod\ p^k)$. Let $b \equiv ap^{k-1} (mod\ p^k)$. 
We have $bh \equiv ap^{k-1}h\equiv ap^{k-1}h' \equiv ah'p^{k-1} \equiv p^{k-1} (mod\ p^k)$.
\end{proof}

Now we analyze the query complexity of Algorithm \ref{al:main1}. In the following, since Step \ref{main1_query} is the only step to makes queries, it suffices to compute the values of $|W_1|, |W_2|$ in Step \ref{main1_query}. Similarly, we discuss the following two cases:

i) $H$ is nontrivial. Since $H_i = \{h\in H|I_h \le i\}$, we have $H_i = H \cap G_i$. And because $H_i \subseteq H$ and for any $h_1,h_2\in H_i$ we have $h_1+h_2 \in H_i$, $H_i$ is a subgroup of $H$. 
Since $d = \min_{h \in H} I_h$, we have $|H_{d-1}| =1$. By Lemma \ref{lemma:only}, we have $|H| \le \prod_{i=d}^l p_i^{k_i}$. Additionally, Algorithm \ref{al:main1} will go over in the first $d$ iterations during the above correctness proof. Since $G_i = \{0\} \times \mathbb{Z}_{p_1^{k_1}} \times \cdots \mathbb{Z}_{p_i^{k_i}}$, we have $i \le \log |G_i|$ and $|G_{i+1}|/|G_{i}| = p_{i+1}^{k_{i+1}} \ge 2$. Thus, 
\begin{equation}
\label{eqratio}
\begin{aligned}
\sum_{m=0}^i \sqrt{|G_i|} \le \sqrt{|G_i|}(1+\frac{1}{\sqrt{2}}+\cdots \frac{1}{\sqrt{2^{i-1}}}) \le \frac{1}{1-1/\sqrt{2}}\sqrt{|G_i|}.  
\end{aligned}
\end{equation}
$W_1,W_2$ are updated by calling Algorithm \ref{al:subroutine} in Step \ref{main1_update1}. Note that after each iteration, $W_1$ is non-decrease. Additionally, $|W_1|\le 2( \sqrt{|G_{i}|}+1), |W_2| \le \sqrt{|G_{i}|}+1$ after running the $i$-th iteration. Thus, 
the query complexity of Algorithm \ref{al:main1} is at most
$$
\begin{aligned}
& 2(\sqrt{|G_{d-1}|} +1) + \sum_{i=0}^{d-1} (\sqrt{|G_{i}|}+1) \\ &\le 2(\sqrt{|G_{d-1}|}+1)+\frac{1}{1-1/\sqrt{2}}\sqrt{|G_{d-1}|}+d \\
&= O(\sqrt{|G_{d-1}|}).
\end{aligned}
$$ 

Since $|G| = \prod_{i=1}^l p_i^{k_i}, |G_{d-1}| = \prod_{i=1}^{d-1} p_i^{k_i}$, we have $|G_{d-1}| \le \frac{|G|}{|H|}$. Thus, the query complexity of Algorithm \ref{al:main1} is $O(\sqrt{\frac{|G|}{|H|}})$.

ii) $H$ is trivial. In this case, Algorithm \ref{al:main1} will perform all $l$ iterations. Similiar to i), we can obtain the query complexity is at most $O(\sqrt{|G_{l-1}|})$. Since $|H| = 1$, we have $O(\sqrt{|G_{l-1}|}) = O(\sqrt{|G|}) = O(\sqrt{\frac{|G|}{|H|}})$.
\begin{Lemma}\label{lemma:only}
If $H_i \neq H_{i-1}$, then there exist $0 \le j < k_i, h \in H_i$ such that $h_i = p_i^j$ and $H_i = H_{i-1}+\langle h \rangle$. Additionally, $|H_i| \le p_i^{k_i-j} |H_{i-1}|$.
\end{Lemma}
\begin{proof}
Suppose $j$ is the minimum integer such that $h_i = cp_i^j$ for any $h \in H_i$. Then there exists $h$ such that $h_i = p_i^j$. For any $h' \in H_i\setminus H_{i-1}$, w.l.o.g., let $h'_i = cp_i^j$. If $h' \neq H_{i-1} + \langle h \rangle$, then $h' - ch \notin H_{i-1}$. On the other hand, since $h' \in H, ch \in H$, we have $h'-ch \in H$. Additionally, $h'-ch \in G_{i-1}$. Thus, $h'-ch \in G_{i-1} \cap H = H_{i-1}$, which leads to a contradiction. Thus, for any $h' \in H_i\setminus H_{i-1}$, we have $h' = H_i + \langle h \rangle$. Thus, $$
H_i = H_{i-1} + \langle h \rangle =\{H_{i-1}+ \alpha h| \alpha \in \mathbb{Z}_{p_i^{k_i}}\},
$$
Since $H_{i-1} = \{a \in H|I_a \le i-1\}$ and $h_i = p_i^j$, for any $\alpha$ satisfying  $ p_i^{k_i-j}| \alpha$, we have $\alpha h \in H_{i-1}$. Thus, 
$H_{i-1}+\alpha_1 h = H_{i-1}+\alpha_2 h$ iff $p_i^{k_i-j} | \alpha_1-\alpha_2$. Hence, 
$$
\begin{aligned}
H_i &= H_{i-1}+\langle h \rangle \\
&= \{H_{i-1}+\alpha h|\alpha \in \mathbb{Z}_{p_i^{k_i}}\} \\
&= \{H_{i-1}+\alpha h|\alpha \in \mathbb{Z}_{p_i^{k_i-j}}\}, \\
\end{aligned}
$$ 
which implies $|H_i| \le p_i^{k_i-j} |H_{i-1}|$.
\end{proof}

\subsection{Identification version}\label{sec:general_case_identification}
We 
give Algorithm \ref{al:main2} to solve the identification version of $\mathsf{HSP}$. Similar to Algorithm \ref{al:main1}, the idea of Algorithm \ref{al:main2} is to find $H_i$ for $1 \le i \le l$ by generating pairs. Meanwhile, in each iteration, the algorithm updates $V$ satisfying that $H_i +V = G_i$ to serve for the next iteration. Finally, the algorithm returns $H_l = H$.
\begin{algorithm}
	\caption{The algorithm for the identification version}
	\label{al:main2}
	\begin{algorithmic}[1]
		\REQUIRE $G = \{0\} \times \mathbb{Z}_{p_1^{k_1}} \times \mathbb{Z}_{p_2^{k_2}} \times \cdots \times \mathbb{Z}_{p_l^{k_l}}$, where $p_i$'s are primes and $k_i \ge 1$ for any $i$;
		\ENSURE The hidden subgroup $H$;
		\STATE $V = W_1 = W_2 = H' = \{0\}, r = 0$; \label{main2_init}
		\FOR{$i=1 \to l$} \label{main2_for1}
			\STATE Query all not queried elements in $W_1 \times \{0\}$; \label{main2_query1}
			\STATE $H' \leftarrow H' \times \{0\}$, $t_i = k_i$; \label{main2_init_t}
			\FOR{$j=0 \to k_i-1$}\label{main2_for2}
			\STATE Query the elements in $W_2 \times \{p_i^j\}$;\label{main2_query2}
				\IF{there exist $x \in W_1 \times \{0\}$, $y \in W_2 \times \{p_i^j\}$ such that $f(x) = f(y)$} \label{main2_if}
				\STATE $H' \leftarrow H'+\langle x-y \rangle$, $t_i = j$; \label{main2_update_H}
					\IF{$j = 0$}
		\STATE $r = r+1$; \label{main2_update_r}
				\ENDIF
				\STATE \textbf{break};
				\ENDIF
			\ENDFOR
			\STATE $V \leftarrow V \times \mathbb{Z}_{p_i^{t_i}}$; \label{main2_update_V}
			\STATE $(W_1,W_2) \leftarrow \text{findPair}(V, \max\{1,r\}$); \label{main2_update_W}
		 \ENDFOR
		 \RETURN $H'$;
		\end{algorithmic}
	\end{algorithm}

In the following, we first describe the process of Algorithm \ref{al:main2}, and then analyze its correctness and query complexity. In Step \ref{main2_init}, we first initialize $V = W_1= W_2 = H' = \{0\}$ and $r =0$. Then we jump into the outer loop. We first query all the elements not queried  in $W_1 \times \{0\}$ in Step \ref{main2_query1} and initialize $t_i = k_i$ in Step \ref{main2_init_t}.  
Then we jump into the inner loop. In Step \ref{main2_query2}, we query the elements in $W_2 \times \{p_i^j\}$. If there exist some collisions, then we do the following things: i) update $H$ and set $t_i=j$; ii) if $j = 0$, let $r=r+1$. iii) jump out the inner loop. Then we update $V,W_1,W_2$ in Step \ref{main2_update_V} and Step \ref{main2_update_W}.  
 
 The correctness analysis of Algorithm \ref{al:main2} is as follows. Note that we always have $V = W_1-W_2$ in the algorithm procedure by calling Algorithm \ref{al:subroutine}. In the following, for the sake of clarity, let $V_0,H'_0,r_0$ be the initial values of $V,H',r$ in Algorithm \ref{al:main2}, and $V_i,H'_i,r_i$ represent the values of $V,H',r$ after running the $i$-th iteration of the outer loop ($1 \le i \le l$). Let $V,r$ be the final values, i.e., $V= V_l$ and $r = r_l$. 

We first prove $V_i + H'_i = G_i$ by induction. i) $i= 0$, $H'_i = V_i = G_i = \{0\}$, so $V_i + H'_i = G_i$. ii) suppose $V_m+H'_m = G_m$ for any $m < i$. If we do not find collisions in the inner loop, then $H'_i = H'_{i-1} \times \{0\}$, $V_i = V_{i-1} \times \mathbb{Z}_{p_i^{k_i}}$, and thus 
$$
H'_i + V_i = G_{i-1}\times \mathbb{Z}_{p_i^{k_i}} = G_i.
$$ 
If there exist collisions in the $j$-th iteration of the inner loop, since $x \in W_1 \times \{0\}, y\in W_1 \times \{p_i^j\}$ and $W_1-W_2 = V_{i-1}$, we have $x-y \in V_{i-1} \times \{p_i^j\}$. Thus, $H'_i = H'_{i-1} \times \{0\}+ \langle x-y \rangle$, $V_i = V_{i-1}\times \mathbb{Z}_{p_i^j}$. Since $\langle p_i^j \rangle + \mathbb{Z}_{p_i^j} \equiv \mathbb{Z}_{p_i^{k_i}} (\text{mod } p_i^{k_i})$, we have $H'_i+V_i = G_{i-1} \times \mathbb{Z}_{p_i^{k_i}} = G_i$.
 

Then we will prove $H'_i = H_i$ by induction. 
First, $i=0$, $H'_i = \{0\} = H_i$. Second, suppose $H'_m = H_{m}$ for any $m < i$. Then we go into the $i$-th iteration of the outer loop. i) If we do not find collisions in the inner loop, then $H'_i = H'_{i-1} = H_{i-1}$. By Lemma \ref{lemma:only}, if $H_i \neq H_{i-1}$, then there exists $h \in H_i$ such that $h_i = p_i^j$ for some $j$ and $H_i = H_{i-1} + \langle h \rangle$. Additionally, since $W_1 \times \{0\}-W_2 \times \{p_i^j\} = V_{i-1} \times \{p_i^j\}$ in the $j$-th iteration of the inner loop, 
we have $v \notin H$ for any $v \in V_{i-1} \times \{p_i^j\}$. Due to 
$$
H_{i-1}+V_{i-1} = H'_{i-1}+V_{i-1} = G_{i-1},
$$
for any $g \in G_{i-1} \times \{p_i^j\}$, there exist $h \in H_{i-1},v \in V_{i-1} \times \{p_i^j\}$ such that $g = h+v$. Since $h \in H, v \notin H$, we have $g \notin H$. Thus, there exists no such $h \in H_i$ that $h_i = p_i^j$, which implies $H_i = H_{i-1}$, i.e., $H'_i = H_i$.
 ii) If we find collisions in the $j$-th iteration of the inner loop, then let $h = x-y = (0,h_1,...,h_{i-1},p_i^j,0,...,0)$. Since $H'_{i} = H'_{i-1}+\langle h \rangle$ and $H'_{i-1} = H_{i-1}$, it suffices to prove $H_{i-1}+\langle h \rangle = H_i$. 
For $h' \in H_i/H_{i-1}$, suppose there exists no integer $c$ such that $h'_i = cp_i^j$. Without loss of generality, we assume $h'_i = cp_i^{j'}$, where $j' < j$ and $gcd(c,p_i) = 1$. Then there exists an integer $x$ such that $xh' \equiv 1 (mod\ p_i^{k_i})$, which means $(xh')_i = p_i^{j'}$. Thus, we will find collisions in the $j'$-th iteration, which leads to a contradiction. Hence, for any $h' \in H_i/H_{i-1}$, there exists an integer $c$ such that $h'_i = cp_i^{j}$. 
By the proof of Lemma \ref{lemma:only}, we have $H_i = H_{i-1}+\langle h \rangle$. Thus, $H'_i = H_i$.
 Totally, after $l$ iterations, Algorithm \ref{al:main2} returns $H'_l = H_l = H$, i.e., Algorithm \ref{al:main2} returns the correct result.

 Next, we analyze the query complexity of Algorithm \ref{al:main2}. Note that Step \ref{main2_query1} and \ref{main2_query2} are the only steps to make queries. 
For $1 \le i \le l$, the running time of Step \ref{main2_query2} in the $i$-th iteration is at most $t_i+1$; let $a_i = |W_1|$ and $b_i = |W_2|$ when running the $i$-th iteration of the outer loop. Let $a_0 = 0$.  
It is worth noting that $W_1$ is non-decrease after each iteration. Since $W_1$ and $W_1 \times \{0\}$ are actually the same set, we essentially query all the elements not queried  in $W_1$ in Step \ref{main2_query1} each time. Thus, in the $i$-th iteration, the number of queries in Step \ref{main2_query1} is $a_{i}-a_{i-1}$; also, the number of queries in Step \ref{main2_query2} is at most $(t_{i}+1) b_{i}$. Thus, the total number of queries is at most $a_{l}+ \sum_{i=1}^l (t_{i}+1) b_{i}$. 


In the following, we first prove $|H_i| \le p_i^{k_i-t_i} |H_{i-1}|$. i) If $H_i = H_{i-1}$, then $t_i = k_i$ and thus $|H_i| = p_i^{k_i-t_i}|H_{i-1}|$. ii) If $H_i \neq H_{i-1}$, then there exists $h$ such that $h_i = p_i^{t_i}$ and $H_i = H_{i-1}+\langle h \rangle$. By Lemma \ref{lemma:only}, we have $|H_i| \le p_i^{k_i-t_i} |H_{i-1}|$. Hence,
$|H_i| \le p_i^{k_i-t_i} |H_{i-1}| \le \prod_{m=1}^i p_m^{k_m-t_m}$. Since $|V_i| = \prod_{m=1}^i p_m^{t_m}$, we hace $|H_i|\cdot |V_i| = \prod_{m=1}^i p_m^{k_m} = |G_i|$. Thus, $|V| = |V_l| = \frac{|G_l|}{|H_l|} = \frac{|G|}{|H|}$. 

Let $S_1 = \{i| r_{i} = r_{i-1}, 1\le i \le l\}$ and $S_2 = \{i | r_{i} = r_{i-1}+1\}$. Then $|S_1| = l-r$ and $|S_2| = r$. Moreover, for $i \in S_1$, we have $t_i \ge 1$ and 
$$
b_i \le \lceil\sqrt{|V_i|/\max\{1,r\}}\rceil \le \lceil\sqrt{V_i}\rceil \le \sqrt{V_i}+1;
$$
for $i \in S_2$, we have $t_i = 0$.  
Similar to Eq. \eqref{eqratio}, we have $\sum_{m=1}^i \sqrt{|V_m|} \le \frac{1}{1-1/\sqrt{2}}\sqrt{|V_i|}$. Since $t_i^2 \le 2^{t_i+2} \le 4\cdot p_i^{t_i}$ for $t_i \ge 1$, we have $t_i \le 2 \frac{\sqrt{|V_i|}}{\sqrt{|V_{i-1}|}}$. So $t_i \sqrt{|V_{i-1}|} \le 2\sqrt{|V_i|}$.
Since
$$
|V| = \prod_{i=1}^l p_m^{t_m}, 
$$
we have
$$
\sum_{i=1}^l t_i \le \sum_{i=1}^l t_m \log p_m =
\log |V|,
$$
which means $\sum_{i=1}^l t_i \le \log_p |V|$.
Moreover, since 
$$
|V| = \prod_{i=1}^l p_m^{t_m} \ge \prod_{i\in S_1}p_m^{t_m} \ge \prod_{i\in S_1} p_m \ge 2^{l-r},
$$
we have $l-r \le \log |V|$. Thus, we have 
$$
\begin{aligned}
\sum_{i \in S_1}(t_{i}+1)b_{i} 
&\le \sum_{i \in S_1}(t_{i}+1)(\sqrt{|V_{i-1}|}+1) \\
&= \sum_{i=1}^l (t_{i}\sqrt{|V_{i-1}|}+\sqrt{|V_{i-1}|})+\sum_{i=1}^l t_{i+1}+l-r \\
&\le \sum_{i=1}^l (2\sqrt{|V_i|}+\sqrt{|V_{i-1}|})+2\log |V| \\
 &\le \sum_{i=1}^l 3 \sqrt{|V_i|}+2\log |V| \\
 &\le \frac{3}{1-1/\sqrt{2}}\sqrt{|V|}+2\log |V| \\
 &= O(\sqrt{|V|}),
\end{aligned}
$$
$$
\begin{aligned}
\sum_{i \in S_2}(t_i+1)b_{i} &= \sum_{i \in S_2}b_{i} \\
&\le \sum_{i \in S_2}(\sqrt{\frac{\sqrt{|V|}}{\max\{1,r_{i-1}\}}}+1) \\
&< \sqrt{|V|}(1+\sum_{i=1}^{r-1} \frac{1}{\sqrt{i}})+r \\
&< \sqrt{|V|}(1+\int_{0}^{r-1}\frac{1}{\sqrt{x}}dx)+r \\
&<2\sqrt{|V|\cdot r}+r \\
&= O(r+\sqrt{|V|\cdot r}).
\end{aligned}
$$



Since $a_{l} \le 2\lceil \sqrt{|V|\cdot r}\rceil$, the total number of queries is $O(r+\sqrt{|V|\cdot r})$. Since $|H| \ge \prod_{i\in S_2} p_i^{k_i} \ge 2^r$, we have $r = O(\log |H|)$. Since $|V| = \frac{|G|}{|H|}$, the total number of queries is $O(\sqrt{\frac{|G|}{|H|}\log |H|}+\log |H|)$.
Finally, Theorem \ref{th:general_case} is implied by the above analysis of Algorithm \ref{al:main1} and \ref{al:main2}.

 



\section{A general lower bound on the query complexity of $\mathsf{HSP}$}\label{sec:lower_bound}
In this section, we give a general lower bound for $\mathsf{HSP}$ by proving Theorem \ref{th:lower_bound}.

\begin{proof}[Proof of Theorem \ref{th:lower_bound}]
We consider the identification version of the hidden subgroup problem with additional promise $(\mathsf{HSP}^+)$. In $\mathsf{HSP}^+$, suppose we have known the size of $H$ and $|X| = \frac{|G|}{|H|}$. Since any algorithm solving $\mathsf{HSP}$ also can solve $\mathsf{HSP}^+$, it suffices to obtain the lower bound on the query complexity of $\mathsf{HSP}^+$. 
Since $\mathsf{HSP}^+$ has at least $|\mathcal{H}|$ possible outputs, where $|\mathcal{H}| = \inbrace{H' \le G||H'| = |H|}$, any decision tree solving the problem must have at least $\Omega(|\mathcal{H}|)$ leaves. Since every query has at most $\frac{|G|}{|H|}$ possible answers, the depth of the decision tree must be at least $\Omega\inparen{\log_{\frac{|G|}{|H|}} |\mathcal{H}|} = \Omega\inparen{\frac{\log |\mathcal{H}|}{\log \frac{|G|}{|H|}}}$.
\end{proof}

\section{Conclusion}
\label{sec:Conclusion}

In this paper, we considered deterministic algorithms to solve the hidden subgroup problem ($\mathsf{HSP}$). First, 
for $\mathsf{HSP}$ over $G = {\mathbb{Z}}_{p^k}$ with $p$ being a prime, we presented a deterministic algorithm with $2$ queries to decide whether $H$ is trivial, and an algorithm with 
$O(\log \frac{|G|}{|H|})$ queries to identify $H$.
Second, For $\mathsf{HSP}$ over a general finite Abelian group $G$, we devised an algorithm with $O(\sqrt{\frac{|G|}{|H|}})$ queries to decide the triviality of $H$ and an algorithm to identify $H$ with $O(\sqrt{\frac{|G|}{|H|}\log |H|}+\log |H|)$ queries. Third, a general lower bound on the query complexity of $\mathsf{HSP}$ was given, which leads to the observation that there exists an instance of $\mathsf{HSP}$ that needs $\omega(\sqrt{\frac{|G|}{|H|}})$ queries to identify $H$.  

These results not only answer the open problem proposed by Nayak \cite{Nayak2021}, but also extend the main results of Ref. \cite{Ye2019query}.
Furthermore, it is worth exploring whether there exist similar deterministic algorithms for the non-Abelian hidden subgroup problem.  
\bibliography{GSP}
\end{document}